\documentclass[onecolumn]{IEEEtran}
\usepackage{amsthm}
\usepackage{times,amssymb,amsmath,amsfonts,float,nicefrac,color,bbm,mathrsfs,caption,float}
\usepackage{enumerate,multirow,caption,tikz,graphicx,enumitem,soul}
\usepackage[mathscr]{eucal}
\usepackage{sidecap}
\usepackage{algpseudocode}
\usepackage{verbatim}
\usepackage{epstopdf}
\usepackage{textcomp}
\usetikzlibrary{shapes,arrows}
\usepackage{mathabx}
\usepackage[noadjust]{cite}
\usepackage{booktabs}
\usepackage[top=1in,bottom=1in,left=1.2in,right=1.2in]{geometry}
\allowdisplaybreaks

\usepackage[font={it}]{caption}
\usepackage[margin=0.05in]{subcaption}
\usepackage[bottom]{footmisc}
\usepackage[normalem]{ulem}

\usepackage[ruled,vlined,linesnumbered]{algorithm2e}

\usepackage{hyperref} 

\newcommand\nc\newcommand
\nc\bfa{{\boldsymbol a}}\nc\bfA{{\boldsymbol A}}\nc\cA{{\mathscr A}}
\nc\bfb{{\boldsymbol b}}\nc\bfB{{\boldsymbol B}}\nc\cB{{\mathscr B}}
\nc\bfc{{\boldsymbol c}}\nc\bfC{{\boldsymbol C}}\nc\cC{{\mathscr C}}
\nc\bfd{{\boldsymbol d}}\nc\bfD{{\boldsymbol D}}\nc\cD{{\mathscr D}}
\nc\bfe{{\boldsymbol e}}\nc\bfE{{\boldsymbol E}}\nc\cE{{\mathscr E}}
\nc\bff{{\boldsymbol f}}\nc\bfF{{\boldsymbol F}}\nc\cF{{\mathscr F}}
\nc\bfg{{\boldsymbol g}}\nc\bfG{{\boldsymbol G}}\nc\cG{{\mathscr G}}
\nc\bfh{{\boldsymbol h}}\nc\bfH{{\boldsymbol H}}\nc\cH{{\mathscr H}}
\nc\bfi{{\boldsymbol i}}\nc\bfI{{\boldsymbol I}}\nc\cI{{\mathcal I}}
\nc\bfj{{\boldsymbol j}}\nc\bfJ{{\boldsymbol J}}\nc\cJ{{\mathscr J}}
\nc\bfk{{\boldsymbol k}}\nc\bfK{{\boldsymbol K}}\nc\cK{{\mathscr K}}
\nc\bfl{{\boldsymbol l}}\nc\bfL{{\boldsymbol L}}\nc\cL{{\mathscr L}}
\nc\bfm{{\boldsymbol m}}\nc\bfM{{\boldsymbol M}}\nc{\cM}{{\mathscr M}}
\nc\bfn{{\boldsymbol n}}\nc\bfN{{\boldsymbol N}}\nc\cN{{\mathscr N}}
\nc\bfo{{\boldsymbol o}}\nc\bfO{{\boldsymbol O}}\nc\cO{{\mathscr O}}
\nc\bfp{{\boldsymbol p}}\nc\bfP{{\boldsymbol P}}\nc\cP{{\mathscr P}}\nc\eP{{\EuScriptP}}\nc\fP{{\mathfrak P}}
\nc\bfq{{\boldsymbol q}}\nc\bfQ{{\boldsymbol Q}}\nc\cQ{{\mathscr Q}}
\nc\bfr{{\boldsymbol r}}\nc\bfR{{\boldsymbol R}}\nc\cR{{\mathscr R}}
\nc\bfs{{\boldsymbol s}}\nc\bfS{{\boldsymbol S}}\nc\cS{{\mathscr S}}
\nc\bft{{\boldsymbol t}}\nc\bfT{{\boldsymbol T}}\nc\cT{{\mathscr T}}
\nc\bfu{{\boldsymbol u}}\nc\bfU{{\boldsymbol U}}\nc\cU{{\mathscr U}}
\nc\bfv{{\boldsymbol v}}\nc\bfV{{\boldsymbol V}}\nc\cV{{\mathscr V}}
\nc\bfw{{\boldsymbol w}}\nc\bfW{{\boldsymbol W}}\nc\cW{{\mathscr W}}
\nc\bfx{{\boldsymbol x}}\nc\bfX{{\boldsymbol X}}\nc\cX{{\mathscr X}}
\nc\bfy{{\boldsymbol y}}\nc\bfY{{\boldsymbol Y}}\nc\cY{{\mathscr Y}}
\nc\bfz{{\boldsymbol z}}\nc\bfZ{{\boldsymbol Z}}\nc\cZ{{\mathscr Z}}

\newtheorem{theorem}{Theorem}
\newtheorem{lemma}[theorem]{Lemma}
\newtheorem{claim}[theorem]{Claim}

\newtheorem{proposition}[theorem]{Proposition}

\newtheorem{corollary}[theorem]{Corollary}
\newtheorem{definition}{Definition}

\newtheorem{construction}{Construction}[section]

\theoremstyle{remark}

\DeclareMathOperator{\Span}{Span}
\newcommand{\ff}{{\mathbb F}}

\begin{document}
	
		\title{Convolutional codes with a maximum distance profile based on skew polynomials}
	
	\author{\IEEEauthorblockN{Zitan Chen}}
	\maketitle	
	
	{\renewcommand{\thefootnote}{}\footnotetext{
			
			\vspace{-.2in}
			
			\noindent\rule{1.5in}{.4pt}

			{
				
				Z. Chen is with SSE and FNii, The Chinese University of Hong Kong, Shenzhen, China. Email: chenztan@cuhk.edu.cn
				
			}
		}
	}
	\renewcommand{\thefootnote}{\arabic{footnote}}
	\setcounter{footnote}{0}

	\begin{abstract} We construct a family of $(n,k)$ convolutional codes with degree $\delta\in\{k,n-k\}$ that have a maximum distance profile. The field size required for our construction is $\Theta(n^{2\delta})$, which improves upon the known constructions of convolutional codes with a maximum distance profile. Our construction is based on the theory of skew polynomials.
	\end{abstract}

\section{Introduction}
Construction of codes with good distance has been central to the study of error-correcting codes. While there are quite a number of algebraic constructions of block codes with good distance properties, very few algebraic constructions are known for convolutional codes with good distance properties. In fact, there are several important distance measures for convolutional codes, suitable for assessing their error correction capability in various scenarios. In particular, \emph{free distance} and \emph{column distance} \cite{costello1969construction} are two of the most notable distance measures for convolutional codes. 

Free distance is arguably the most well-known notion of distance measure for convolutional codes. An upper bound for the free distance of convolutional codes was presented in \cite{rosenthal1999maximum}, which generalizes the Singleton bound for block codes, and thus convolutional codes with free distance attaining this upper bound are called maximum distance separable (MDS) convolutional codes. The existence of MDS convolutional codes is shown in \cite{rosenthal1999maximum} using methods from algebraic geometry. Relying on the intrinsic connection between quasi-cyclic block codes and convolutional codes \cite{justesen1973new,massey1973polynomial}, the authors of \cite{smarandache2001constructions} presented the first concrete construction of MDS convolutional code under certain restriction of the code parameters. A few other constructions of MDS convolutional codes for different regimes of the code parameters appeared later in \cite{gluesing2006class}, \cite{plaza2013construction}, \cite{julia7constructions}. 

Although both notions of free distance and column distance were born at the same time \cite{costello1969construction}, the notion of column distance appears to be less known. Informally speaking, column distance measures the Hamming distance of a convolutional code truncated at certain time instant. Thus, a convolutional code is associated with a sequence of column distances at different time instants, and such a sequence is termed the \emph{distance profile} of the convolutional code. In view of this, the distance profile of a convolutional code is particularly helpful for evaluating the error correction capability of the code when it is used for streaming over erasure channels \cite{tomas2012decoding}, \cite{mahmood2016convolutional}. A comprehensive study of convolutional codes with the largest possible column distances was presented in \cite{gluesing2006strongly}. Since a truncated convolutional code is a block code, the column distances of the code satisfy the corresponding Singleton bound for block codes. Moreover, the column distances can not exceed the free distance of the code since the latter is derived with no restriction on the time instant. Clearly, not every truncated code of a convolutional code can be an MDS block code. In other words, not every column distance of a convolutional code can attain the corresponding block code Singleton bound. Convolutional codes whose column distances attain the block code Singleton bound for the maximum number of times are said to have a maximum distance profile (MDP), and such codes are often referred to as MDP convolutional codes.

Algebraic constructions of MDP convolutional codes are scarce in the literature. In \cite{gluesing2006strongly}, the authors presented the first explicit construction of MDP convolutional codes over a large finite field of large characteristic. The idea that is instrumental in this construction is to first construct lower triangular Toeplitz supperregular matrices\footnote{Roughly speaking, a lower triangular Toeplitz matrix is called supperregular if its minors that are not trivially zero are nonzero. We refer the readers to \cite{gluesing2006strongly} for a detailed discussion.} and then utilize them as the building blocks for the construction of MDP convolutional codes. The same idea was also explored in \cite{almeida2013new} to construct another class of MDP convolutional codes over a large finite field of arbitrary characteristic. A third general class of explicit MDP convolutional codes was recently presented in \cite{alfarano2020weighted}, which uses generalized Reed-Solomon (RS) codes as the building blocks. These three families of codes are the only known explicit algebraic constructions of MDP convolutional codes hitherto but all of them require a finite field of gigantic size for the codes to be constructed. Specifically, to construct an MDP convolutional code with rate $R=k/n$ and degree\footnote{The degree of a convolutional code determines the implementation complexity of the code and is arguably one of the most fundamental parameters of the code. See Section~\ref{sec:pre-conv} for a formal definition.} $\delta$, a finite field $\ff$ of size 
$
	|\ff| = 2^{\Theta((R^{-1}\delta+n)^2)}
$ suffices for the construction in \cite{gluesing2006strongly}.
As for the construction in \cite{almeida2013new}, the needed size of the finite field is
$
	|\ff| = 2^{\Theta(2^{R^{-1}\delta+n})}
$.
For the most recent construction in \cite{alfarano2020weighted}, it requires 
$
	|\ff| = 2^{\Theta((\frac{\delta^3}{Rn}+\delta Rn)\log n)}
$, which is comparable to that of \cite{gluesing2006strongly}. Constructing explicit MDP convolutional codes over small finite fields has been challenging and it was conjectured in \cite{hutchinson2008superregular} that a finite field of size
$
	|\ff|= 2^{\Theta(R^{-1}\delta+n)}
$ should suffice for constructing of MDP convolutional codes with rate $R=k/n$ and degree $\delta$.

In this paper we present an explicit construction of MDP convolutional codes over a finite field of size $2^{\Theta(\delta\log n)}$, which improves over all the known explicit algebraic constructions. More precisely, we show that for any rate $R\in(0,1)$ such that $R\neq 1/2$, there exists an explicit MDP convolutional code with rate $R$ and degree $\delta\in\{k,n-k\}$ over a finite field of size $\Theta(n^{2\delta})$. The idea behind our construction differs from the existing constructions and is based on the theory of skew polynomials \cite{lam1988vandermonde}, \cite{lieb2020convolutional}, which has seen important applications in other areas of coding theory such as distributed storage \cite{martinez2019universal}, \cite{cai2021construction}.

The paper is organized as follows. Section~\ref{sec:pre} provides a minimum background on convolutional codes and skew polynomials. Section~\ref{sec:1mds} presents our construction of MDP convolutional codes. Finally, we conclude this paper with a few remarks in Section~\ref{sec:con}.

\section{Preliminaries}
\label{sec:pre}
\subsection{Convolutional codes}
\label{sec:pre-conv}

Let us first introduce basic concepts and requisite terminology of convolutional codes. We shall treat convolutional codes as submodules instead of subspaces since there is only marginal difference between this treatment and the classical approach \cite{johannesson2015fundamentals} of taking convolutional codes as vector subspaces. We refer the readers to \cite{lieb2020convolutional} for a recent survey of convolutional codes based on the module theoretic approach.

Let $\ff$ be a finite field and $\ff[D]$ be the ring of polynomials with indeterminate $D$ and coefficients in $\ff$.

\begin{definition}
	An $(n,k)$ convolutional code over $\ff$ is an $\ff[D]$-submodule $\cC\subset \ff[D]^n$ of rank $k$. A $k\times n$ polynomial matrix $G(D) = (g_{ij}(D))$ over $\ff[D]$ is called a generator matrix of the code $\cC$ if 
	\begin{align*}
		\cC= \{u(D)G(D)\mid u(D)\in \ff[D]^k\}.
	\end{align*}
	 An $(n-k)\times n$ polynomial matrix $H(D)=(h_{ij}(D))$ over $\ff[D]$ is called a parity check matrix of the code $\cC$ if 
	\begin{align*}
		\cC = \{v(D)\in\ff[D]^n\mid H(D)v(D)=0\}.
	\end{align*}
\end{definition} 

A $k\times n$ polynomial matrix $G(D)$ is called \emph{basic} if it has a polynomial right inverse. The \emph{constraint length for the $i$th input} or the \emph{$i$th row degree} of the matrix $G(D)$ is defined to be 
\begin{align*}
	\nu_i=\max_{1\leq j\leq n}\{\deg g_{ij}(D)\}, \quad i=1,\ldots,k.
\end{align*}
Furthermore, the \emph{memory} $m$ of the matrix $G(D)$ is defined as 
\begin{align*}
	m=\max_{1\leq i\leq k}\{\nu_i\},
\end{align*}
and the \emph{overall constraint length} of the matrix $G(D)$ is defined to be 
\begin{align*}
	\nu=\sum_{i=1}^{k}\nu_i.
\end{align*}
The matrix $G(D)$ is said to have \emph{generic row degrees} if $\nu_i\in\{m,m-1\}$ for $i=1,\ldots,k$.

A generator matrix $G(D)$ of an $(n,k)$ convolutional code $\cC$ is called \emph{minimal} or \emph{reduced} if its overall constraint length is minimal over all generator matrices of the code $\cC$.
This minimum overall constraint length is called the \emph{degree} $\delta$ of the code $\cC$.\footnote{One can equivalently define the degree of $\cC$ as the highest degree of the $k\times k$ minors of any generator matrix of $\cC$. This is because the overall constraint length of a generator matrix is always at least the highest degree the of the $k\times k$ minors of the generator matrix \cite{mceliece1993general}, and the sets of degrees of $k\times k$ minors of any two generator matrices are the same since any two generator matrices of $\cC$ differs by left multiplication of a unimodular polynomial matrix.} We shall call an $(n,k)$ convolutional code with degree $\delta$ an $(n,k,\delta)$ convolutional code. 

The degree of $(n,k)$ convolutional codes stipulates the smallest number of memory elements needed to realize an $(n,k)$ convolutional code and is closely related to the decoding complexity of the code \cite{johannesson2015fundamentals}. In this regard, given code parameters $n$ and $k$, it is preferable to construct codes with a small degree in practice. However, it is in general not obvious to determine the degree of an $(n,k)$ convolutional code. The following result from \cite{mceliece1993general} provides a criterion for a generator matrix to be minimal, which is helpful in determining the degree of the convolutional code generated by the matrix.
\begin{lemma}\label{le:minimal-matrices}
	Let $G(D)$ be a $k\times n$ matrix over $\ff[D]$ and define the matrix of the highest order coefficients for $G(D)$, denoted by $\bar{G}=(\bar{G}_{ij})$, by
	\begin{align*}
		\bar{G}_{ij}=\mathrm{coeff}_{D^{\nu_i}}g_{ij}(D),
	\end{align*} where $\mathrm{coeff}_{D^{\nu_i}}g_{ij}(D)$ denotes the coefficient of $D^{\nu_i}$ in the polynomial $g_{ij}(D)$. Then $G(D)$ is minimal if and only if $\bar{G}$ has rank $k$.
\end{lemma}

Next, let us briefly discuss distance properties of convolutional codes. To begin with, note that for $v(D)\in \ff[D]^n$, we may write $v(D)=\sum_{j\in\mathbb{N}}v_j D^j\in \ff^n[D]$. The Hamming weight of $v(D)$ is then defined as 
\begin{align*}
	\mathrm{wt}(v(D)) = \sum_{j\in\mathbb{N}}\mathrm{wt}(v_j),
\end{align*} where $\mathrm{wt}(v_j)$ denotes the Hamming weight of the length-$n$ vector $v_j$ over $\ff$. Similarly to the distance of a linear block code, the \emph{free distance} of a convolutional code $\cC$ is defined to be 
\begin{align*}
	d=\min\{\mathrm{wt}(v(D))\mid v(D)\in\cC, v(D)\neq 0\}.
\end{align*}

The free distance was shown in \cite{rosenthal1999maximum} to satisfy the following upper bound that generalizes the Singleton bound for block codes.
\begin{theorem}
	Let $\cC$ be an $(n,k,\delta)$ convolutional code. Then
	\begin{align}
		d\leq (n-k)\left( \left\lfloor\frac{\delta}{k}\right\rfloor+1\right)+\delta+1.\label{eq:sb}
	\end{align}
\end{theorem}
Convolutional codes whose distance attains \eqref{eq:sb} with equality are called MDS convolutional codes.

In addition to the free distance, the \emph{column distance} is another fundamental distance measure for convolutional codes. To define the column distance, let $G(D)=\sum_{i=0}^{m}G_iD^i$ be a generator matrix of memory $m$ for an $(n,k)$ convolutional code $\cC$. Further, let the corresponding semi-infinite matrix $G$ be
\begin{align*}
	G=\begin{pmatrix}
		G_0 & G_1 & \cdots & G_m & & \\
			& G_0 & G_1 & \cdots & G_m & \\
			&	  &	\ddots & \ddots & & \ddots 
	\end{pmatrix}.
\end{align*}
Let $j\geq 0$. Denote the \emph{$j$th truncated generator matrix} by
\begin{align*}
	G_j^c=\begin{pmatrix}
		G_0 & G_1 & \cdots & G_j\\
			& G_0 & \cdots & G_{j-1}\\
			&	& \ddots & \vdots\\
			&	& 		 & G_0
	\end{pmatrix},
\end{align*} where $G_i=0$ for $i>m$. Similarly, let $H(D)$ be a parity check matrix of memory $m'$ for the code $\cC$. The \emph{$j$th truncated parity check matrix} is denoted by
\begin{align*}
	H_j^c=\begin{pmatrix}
	H_0	&	   & 		 & \\
	H_1	& H_0  & 		 & \\
	\vdots	& \vdots & \ddots & \\
	H_j & H_{j-1} & \cdots & H_0
	\end{pmatrix}, 
\end{align*} where $H_i=0$ for $i>m'$.

With the $j$th truncated matrix, the $j$th column distance of a convolutional code can be defined as follows.
\begin{definition}
	Let $G(D)=\sum_{i=0}^{m}G_iD^i$ be a generator matrix of memory $m$ for an $(n,k)$ convolutional code $\cC$ such that $G_0$ has full rank. For $j\geq 0$ the $j$th column distance of the code $\cC$ is given by 
	\begin{align*}
		d_j^c=\min\{\mathrm{wt}\big((u_0,\ldots,u_j)G_j^c\big)\mid u_0\neq 0, u_i\in\ff^k,i=0,\ldots,j\}.
	\end{align*}
\end{definition}
The Singleton bound for block codes implies that for all $j\geq 0$ the column distance satisfies 
\begin{align}
	d_j^c\leq (n-k)(j+1)+1,\label{eq:sb-c}
\end{align} and equality for a given $j$ implies that all the other distances $d_i^c,i\leq j$ also attain their versions of the bound \eqref{eq:sb-c} with equality. The following result, given in \cite{gluesing2006strongly}, characterizes optimal column distances by the determinants of full-size square submatrices of $G_j^c$.

\begin{theorem}\label{thm:cd}
	Let $G(D)$ be a $k\times n$ basic and minimal generator matrix of an $(n,k)$ convolutional code and let $H(D)$ be a $(n-k)\times n$ basic parity check matrix of the code. 
	Then the following are equivalent:
	\begin{enumerate}
		\item $d_j^c=(n-k)(j+1)+1$;
		\item every $k(j+1)\times k(j+1)$ full-size minor of $G_j^c$ formed by the columns with indices $1\leq t_1 < \ldots < t_{k(j+1)}$ where $t_{ks+1}\geq ns+1$ for $s=1,\ldots,j$ is nonzero;
		\item every $(n-k)(j+1)\times (n-k)(j+1)$ full-size minor of $H_j^c$ formed by the columns with indices $1\leq t_1<\ldots <t_{(n-k)(j+1)}$ where $t_{ks}\leq ns$ for $s=1,\ldots,j$ is nonzero.
	\end{enumerate}
\end{theorem}

For brevity, we shall call Item~2 of Theorem~\ref{thm:cd} the \emph{MDP property} of $G_j^c$, and Item~3 of Theorem~\ref{thm:cd} the \emph{MDP property} of $H_j^c$.

While for a $k\times n$ matrix to generate an $(n,k)$ MDS block code, every $k\times k$ submatrix should have a nonzero determinant, the MDP property of $G_j^c$ implies that for an convolutional code to have optimal $j$th column distance, the full-size square submatrices that should have nonzero determinants are only those $k(j+1)\times k(j+1)$ submatrices formed by at most $ks$ columns of the first $ns$ columns of $G_j^c$ for $s=1,\ldots,j$. In fact, the determinant of any other full-size submatrix of $G_j^c$ is always zero.

As there are various distance measures for convolutional codes, their error correction capability can be optimal in a number of different ways. We mention here several families of distance-optimal convolutional codes.
Convolutional codes with a generator matrix of memory $m$ and distance attaining the equality in \eqref{eq:sb-c} for $j=m$ are called $m$-MDS convolutional codes.
Relating the bounds \eqref{eq:sb} and \eqref{eq:sb-c}, \cite{hutchinson2005convolutional} and \cite{gluesing2006strongly} introduced two other types of distance-optimal convolutional codes, described below.

\begin{definition}\label{def:sMDS-MDP}
	Let $\cC$ be an $(n,k,\delta)$ convolutional code. Let $M=\lfloor\frac{\delta}{k}\rfloor+\lceil\frac{\delta}{n-k}\rceil$ and $L=\lfloor\frac{\delta}{k}\rfloor+\lfloor\frac{\delta}{n-k}\rfloor$.
	\begin{enumerate}
		\item The code $\cC$ is said to be strongly-MDS if
		\begin{align*}
			d_M^c=(n-k)\left(\left\lfloor\frac{\delta}{k}\right\rfloor+1\right)+\delta+1.
		\end{align*}
		\item The code $\cC$ is said to be MDP if 
		\begin{align*}
			d_L^c=(n-k)(L+1)+1.
		\end{align*}
	\end{enumerate}
\end{definition}
Clearly, strongly-MDS convolutional codes are also MDS by definition. Moreover, the column distance of strongly-MDS convolutional codes attains the Singleton bound \eqref{eq:sb} at the earliest possible time instant $M$. But the column distance at the time instant prior to $M$ may not attain \eqref{eq:sb-c} . On the contrary, the column distances of MDP convolutional codes attain \eqref{eq:sb-c} for the maximum number of times although the free distance of the codes may not attain \eqref{eq:sb}. We refer the readers to \cite{gluesing2006strongly} for a detailed discussion of the difference between these two families of codes. 

Theorem~\ref{thm:cd} suggests that one way of constructing MDP convolutional codes is to design a minimal and basic generator matrix such that the corresponding truncated matrix has the MDP property. Note that Lemma~\ref{le:minimal-matrices} provides a characterization of minimal matrices. As for basic matrices, 
it is observed in \cite{alfarano2020left} that the assumption of being basic in Theorem~\ref{thm:cd} can be lifted and replaced with a weaker condition that the matrix has generic row degrees. This observation is summarized as follows.

\begin{lemma}\label{le:mdp-generic}
	Let $G(D)$ be a $k\times n$ minimal generator matrix with row degree $\nu_i\in\{m,m-1\}$ and let $\delta=\sum_{i=1}^{k}\nu_i$. If $G_L^c$ has the MDP property then $G(D)$ is basic and it generates an $(n,k,\delta)$ MDP convolutional code.
\end{lemma}

Finally, we note a duality result of MDP convolutional codes from \cite{gluesing2006strongly}.
\begin{theorem}\label{thm:mdp-dual}
	An $(n,k,\delta)$ convolutional code $\cC$ over $\ff$ with generator matrix $G(D)$ is MDP if and only if its dual code $\cC^\perp$ with parity check matrix $G(D)$ is an $(n,n-k,\delta)$ MDP convolutional code over $\ff$.
\end{theorem}

\subsection{Skew polynomials}
Skew polynomials satisfy most basic properties of conventional polynomials but the product of skew polynomials is not commutative. We refer the readers to \cite{ore1933theory} for basic properties of skew polynomials over finite fields and \cite{lam1985general,lam1988vandermonde} for the theory of evaluation and interpolation on skew polynomials. The recent papers \cite{martinez2018skew,gopi2020improved} also have an accessible exposition of the theory of skew polynomials. In the following we shall only briefly introduce basic results of skew polynomials that suit the need of this paper.

Let $\ff_{q^t}$ be a finite field of size $q^t$ where $q$ is a prime power. Let $\sigma\colon \ff_{q^t}\to\ff_{q^t}$ be the Frobenius automorphism. Namely, $\sigma(a)=a^q$ for any $a\in\ff_{q^t}$.

\begin{definition}
	Let $\ff_{q^t}[x;\sigma]$ be the ring of skew polynomials with indeterminate in $x$ and coefficients in $\ff_{q^t}$, where addition in $\ff_{q^t}[x;\sigma]$ is coefficient wise and multiplication in $\ff_{q^t}[x;\sigma]$ is distributive and satisfies that for any $a\in\ff_{q^t}$
	\begin{align*}
		xa=\sigma(a)x.
	\end{align*} 
\end{definition}

As the product of skew polynomials is not commutative, evaluating skew polynomials over elements in a finite field is different from evaluating conventional polynomials. The following result, given in \cite{lam1985general,lam1988vandermonde}, provides a simple approach to evaluate skew polynomials.

\begin{definition}
	For any $a\in\ff_{q^t}$, let
	\begin{enumerate}
		\item $N_0(a)=1$;
		\item $N_{i+1}(a)=\sigma(N_i(a))a$ for $i\in\mathbb{N}$.
	\end{enumerate}
\end{definition}

\begin{lemma}\label{le:eval}
	Let $f(x)=\sum_{i=0}f_ix^i\in\ff_{q^t}[x;\sigma]$. Then the evaluation of $f(x)$ at any $a\in\ff_{q^t}$ is given by
	\begin{align*}
		f(a) = \sum_{i}f_iN_i(a).
	\end{align*}
\end{lemma}

Different from conventional polynomials, a skew polynomial $f(x)\in\ff_{q^t}[x;\sigma]$ may have more than $\deg f$ roots in $\ff_{q^t}$. Moreover, the roots may belong to distinct \emph{conjugacy classes} in $\ff_{q^t}$ induced by the field automorphism $\sigma$. The notion of conjugacy classes, introduced in \cite{lam1985general,lam1988vandermonde}, is given below.

\begin{definition}
	Let $a,b\in\ff_{q^t}$. Then $b$ is a conjugate of $a$ with respect to the field automorphism $\sigma$ if there exists $\beta\in\ff_{q^t}^*$ such that 
	\begin{align*}
		b=\sigma(\beta)a\beta^{-1}.
	\end{align*}
	We also call $b$ the $\beta$-conjugate of $a$ with respect to $\sigma$ and write $b={}^\beta a$ for brevity.
\end{definition}

The notion of conjugacy above defines an equivalence relation in $\ff_{q^t}$, and thus we can partition $\ff_{q^t}$ into distinct conjugacy classes. 
For any $a\in\ff_{q^t}$, let us denote by $C_{a}^{\sigma}=\{\sigma(\beta)a\beta^{-1}\mid \beta\in\ff_{q^t}^*\}$ the conjugacy class with representative $a\in\ff_{q^t}$.

\begin{proposition}\label{prop:conjugacy}
	Let $\gamma$ be a primitive element of $\ff_{q^t}$. Then $\{C_{0}^{\sigma},C_{\gamma^0}^{\sigma},C_{\gamma^1}^{\sigma},\ldots,C_{\gamma^{q-2}}^{\sigma}\}$ is a partition of $\ff_{q^t}$. Moreover, $|C_{0}^{\sigma}|=1$ and $|C_{\gamma^i}^{\sigma}|=\frac{q^t-1}{q-1}$ for $i=0,\ldots,q-2$.
\end{proposition}

The set of roots in $\ff_{q^t}$ of a skew polynomial $f(x)\in\ff_{q^t}[x;\sigma]$ has interesting structures closely related to the conjugacy classes of $\ff_{q^t}$. The result below can be found in \cite{gopi2020improved}, \cite{lam1985general,lam1988vandermonde}.

\begin{theorem}\label{thm:skew-dim}
	Let $f(x)\in\ff_{q^t}[x;\sigma]$ be a nonzero skew polynomial and $\Omega$ be the set of roots of $f(x)$ in $\ff_{q^t}$. Further, let $\bigcup_{i}\Omega_i=\Omega$ be the partition of $\Omega$ into conjugacy classes and let $\cS_{i}=\{\beta \mid {}^\beta a_i\in \Omega_i\}\cup\{0\}$ where $a_i$ is some fixed representative in $\Omega_i$. Then $\cS_{i}$ is a vector space over $F_i$ where $F_i=\ff_{q^t}$ if $\Omega_i=C_0^{\sigma}$ and otherwise $F_i=\ff_q$. Moreover, 
	\begin{align*}
		\sum_{i}\dim_{F_i} \cS_{i} \leq \deg f.
	\end{align*}
\end{theorem}

\begin{definition}\label{def:skew-vandermonde}
	Let $k$ be a positive integer and let $\Omega=\{a_1,\ldots,a_n\}\subset \ff_{q^t}$. The $k\times n$ skew Vandermonde matrix with respect to $\Omega$, denoted by $V_k(\Omega)$, is defined to be 
	\begin{align*}
		V_k(\Omega)=\begin{pmatrix}
			N_0(a_1) & N_0(a_2) & \cdots & N_0(a_n)\\
			N_1(a_1) & N_1(a_2) & \cdots & N_1(a_n)\\
			\vdots & \vdots & \ddots & \vdots\\
			N_{k-1}(a_1) & N_{k-1}(a_2) & \cdots & N_{k-1}(a_n)
		\end{pmatrix}.
	\end{align*} 
\end{definition}

The following corollary is a consequence of Theorem~\ref{thm:skew-dim}.

\begin{corollary}\label{coro:skew-vandermonde}
	Let $\Omega$ be an $n$-subset of $\ff_{q^t}$ and $\bigcup_{i}\Omega_i=\Omega$ be the partition of $\Omega$ into conjugacy classes. Let $l_i=|\Omega_i|$ and fix $a_i\in \Omega_i$. Suppose further that $\Omega_i=\{{}^{\beta_{ij}}a_i\mid j=1,\ldots,l_i\}$. Then $V_n(\Omega)$ has full rank if and only if for each $i$ the set $\{\beta_{ij}\mid j=1,\ldots,l_i\}$ is linearly independent over $F_i$ where $F_i=\ff_{q^t}$ if $\Omega_i=C_0^{\sigma}$ and otherwise $F_i=\ff_q$.
\end{corollary}

Note that evaluation of skew polynomials does not preserve linearity. Nevertheless, the following result provides a way to linearize the evaluation of skew polynomials on any conjugacy classes, which was first shown in \cite{lam1985general}.


\begin{lemma}\label{le:skew-linearized}
	Let $f(x)\in\ff_{q^t}[x;\sigma], a\in\ff_{q^t}$, and $\beta\in\ff_{q^t}^*$. Then
	$\cD_{f,a}(\beta):=f({}^\beta a)\beta$ is an $\ff_q$-linear map. In other words, for any $\lambda_1,\lambda_2\in\ff_q$ and $\beta_1,\beta_2\in\ff_{q^t}$, we have $\cD_{f,a}(\lambda_1\beta_1+\lambda_2\beta_2)=\lambda_1\cD_{f,a}(\beta_1)+\lambda_2\cD_{f,a}(\beta_2)$.
\end{lemma}

Proposition~\ref{prop:conjugacy}, Theorem~\ref{thm:skew-dim}, and Lemma~\ref{le:skew-linearized} together imply the corollary below.

\begin{corollary}\label{coro:skew-ker-dim}
		Let $f(x)\in\ff_{q^t}[x;\sigma]$ be a nonzero skew polynomial and let $a\in\ff_{q^t}^*$. Further, let $\ker \cD_{f,a}:=\{\beta\in\ff_{q^t}^*\mid \cD_{f,a}(\beta)=0\}\cup \{0\}$. If $\gamma$ be a primitive element of $\ff_{q^t}$ then
	\begin{align*}
		\sum_{i=0}^{q-2}\dim_{\ff_q} \ker \cD_{f,\gamma^i} \leq \deg f.
	\end{align*}
\end{corollary}

\section{Construction} 
\label{sec:1mds}

In this section we will construct $(n,k)$ MDP convolutional codes with degree $\delta\in\{k,n-k\}$. In the following, we will first present a construction of unit memory convolutional codes whose generator matrix is constructed based on skew Vandermonde matrices.

\begin{construction}\label{con:1} 
Let $t=2k$ and let $q\geq\max\{3,n\}$ be a prime power. Let $\ff_{q^t}$ be a finite field and $\gamma$ be a primitive element of $\ff_{q^t}$.\footnote{By the assumption that $q\geq\max\{3,n\}$, it follows from Proposition~\ref{prop:conjugacy} that there exists at least three conjugacy classes in $\ff_{q^t}$.}
Further, let $\{\lambda_1,\ldots,\lambda_n\}\subset\ff_q$ be $n$ distinct elements and define $\alpha_1,\ldots,\alpha_n,\beta_1,\ldots,\beta_n\in\ff_{q^t}$ using an arbitrary basis for $\ff_{q^t}$ over $\ff_q$ as 
\begin{align}
	\alpha_i&=
	(1,
	\lambda_i,
	\lambda_i^2,
	\ldots,
	\lambda_i^{k-1},
	0,
	\ldots,
	0)
	,\quad i=1,\ldots,n;\label{eq:alpha}\\
	\beta_i&=
	(1,
	\lambda_i,
	\lambda_i^2,
	\ldots,
	\lambda_i^{t-1}
	)
	,\quad i=1,\ldots,n.\label{eq:beta}
\end{align}

Let $\cC$ be an $(n,k)$ convolutional code over $\ff_{q^t}$ with generator matrix $G(D)=G_0+G_1D$ where $G_0$ and $G_1$ are given by
\begin{align}
	G_0=\begin{pmatrix}
		N_0({}^{\alpha_1} 1)\alpha_1 & N_0({}^{\alpha_2} 1)\alpha_2 & \cdots & N_0({}^{\alpha_n} 1)\alpha_n\\
		N_1({}^{\alpha_1} 1)\alpha_1 & N_1({}^{\alpha_2} 1)\alpha_2 & \cdots & N_1({}^{\alpha_n} 1)\alpha_n\\
		\vdots & \vdots & \ddots & \vdots\\
		N_{k-1}({}^{\alpha_1} 1)\alpha_1 & N_{k-1}({}^{\alpha_2} 1)\alpha_2 & \cdots & N_{k-1}({}^{\alpha_n} 1)\alpha_n
	\end{pmatrix},\label{eq:g0}\\
	G_1=\begin{pmatrix}
		N_0({}^{\beta_1} \gamma)\beta_1 & N_0({}^{\beta_2} \gamma)\beta_2 & \cdots & N_0({}^{\beta_n} \gamma)\beta_n\\
		N_1({}^{\beta_1} \gamma)\beta_1 & N_1({}^{\beta_2} \gamma)\beta_2 & \cdots & N_1({}^{\beta_n} \gamma)\beta_n\\
		\vdots & \vdots & \ddots & \vdots\\
		N_{k-1}({}^{\beta_1} \gamma)\beta_1 & N_{k-1}({}^{\beta_2} \gamma)\beta_2 & \cdots & N_{k-1}({}^{\beta_n} \gamma)\beta_n
	\end{pmatrix}. 
\end{align}
\end{construction}

Before proving the code $\cC$ is a unit memory MDP convolutional code, let us first show that its generator matrix $G(D)$ has certain desirable properties.
\begin{proposition}\label{prop:skew-vandermonde-full-rank}
	Any $k\times k$ submatrix of $G_0$ or $G_1$ defined in Construction~\ref{con:1} has full rank.
\end{proposition}
\begin{proof}
	The construction \eqref{eq:alpha} of $\alpha_i\in\ff_{q^t}$ ensures that any $k$ elements of $\{\alpha_i\mid i=1,\ldots,n\}$ are linearly independent over $\ff_q$. By Corollary~\ref{coro:skew-vandermonde}, any $k\times k$ submatrix of the skew Vandermonde matrix $V_k(\{{}^{\alpha_1}1,\ldots,{}^{\alpha_n}1\})$ has full rank.
	Observe that $G_0$ can be obtained by multiplying the $i$th column of $V_k(\{{}^{\alpha_1}1,\ldots,{}^{\alpha_n}1\})$ the nonzero element $\alpha_i$ for all $i=1,\ldots,n$. It follows that any $k\times k$ submatrix of $G_0$ has full rank. Similarly, one can show that, by the construction \eqref{eq:beta} of $\beta_i\in\ff_{q^t}$, any $k\times k$ submatrix of $G_1$ has full rank.
\end{proof}

\begin{proposition}\label{prop:unit-mem-minimal}
	The $k\times n$ generator matrix $G(D)$ defined in Construction~\ref{con:1} is minimal.
\end{proposition}
\begin{proof}
	It is clear that the matrix of the highest order coefficients for $G(D)$ is $\bar{G}=G_1$. Moreover, by Proposition~\ref{prop:skew-vandermonde-full-rank}, $\bar{G}$ has rank $k$. Therefore, by Lemma~\ref{le:minimal-matrices}, $G(D)$ is minimal.
\end{proof}

Now we are ready to prove that $\cC$ is an MDP convolutional code whenever the rate $k/n< 1/2$.

\begin{theorem}\label{thm:m1}
	Let $n>2k$. Then the $(n,k)$ code $\cC$ defined in Construction~\ref{con:1} is an $(n,k,\delta=k)$ unit memory MDP convolutional code over $\ff_{q^t}$. 
\end{theorem}

\begin{proof}
	It is clear that the memory of the generator matrix $G(D)$ is $m=1$ and we have
	\begin{align*}
		G_1^c = \begin{pmatrix}
			G_0 & G_1\\
			& G_0
		\end{pmatrix}.
	\end{align*}

By Proposition~\ref{prop:unit-mem-minimal}, $G(D)$ is minimal, and thus the degree of $\cC$ is $\delta=k$. Since $n>2k$, we have $L=\lfloor\frac{\delta}{k}\rfloor+\lfloor\frac{\delta}{n-k}\rfloor=1$. Note that $G(D)$ has generic row degrees. Therefore, by Lemma~\ref{le:mdp-generic}, it suffices to show that $G_1^c$ has the MDP property. Namely, if $G_1^c$ satisfies Item~2 of Theorem~\ref{thm:cd}, then $\cC$ is an MDP convolutional code with degree $\delta=k$.


Let $A,B\subset\{1,\ldots,n\}$ be such that $|A|\leq k$ and $|B|=2k-|A|$. Let $P$ be the $2k\times 2k$ matrix formed by columns of $G_1^c$ with indices in the set $A\cup (n+B)$ where $n+B$ means adding every element of $B$ by the integer $n$. Further, let $f = (f_0, f_1, \ldots, f_{k-1})\in\ff_{q^t}^k $ and $h = (h_0, h_1, \ldots, h_{k-1})\in \ff_{q^t}^{k}$. We would like to show that $(f,h)P=0$ if and only if $(f,h)=0$. Let $f(x)=\sum_{i=0}^{k-1}f_ix^i\in\ff_{q^t}[x,\sigma],h(x)=\sum_{i=0}^{k-1}h_ix^i\in\ff_{q^t}[x,\sigma]$ be skew polynomials. Then, by Lemma~\ref{le:eval} and Lemma~\ref{le:skew-linearized}, it is equivalent to showing that 
\begin{align}
	\cD_{f,1}(\alpha_u) = 0,\quad u\in A;\label{eq:a}\\
	\cD_{f,\gamma}(\beta_v) + \cD_{h,1}(\alpha_v) = 0,\quad v\in B.\label{eq:b}
\end{align}
if and only if $(f,h)=0$.
For clarity, let us write the matrix $P$ more explicitly as
\begin{align*}
	P=\begin{pmatrix}
		G_{0,A} & G_{1,B}\\
			    & G_{0,B}
	\end{pmatrix},
\end{align*} where $G_{0,A}$ is the submatrix of $G_0$ with column indices in $A$, and $G_{1,B}$ and $G_{0,B}$ are defined similarly. 

Consider the case $|A|=k$. In this case, we have $|B|=k$. By Proposition~\ref{prop:skew-vandermonde-full-rank}, $G_{0,A},G_{1,B},G_{0,B}$ are $k\times k$ matrices of full rank. It follows that $(f,h)P=0$ if and only if $(f,h)=0$.

Consider the case $|A|<k$. In this case, we have $|B|>k$. 
Suppose there exists $(f,h)\neq 0$ such that $(f,h)P=0$.
Suppose further that $h\neq 0$ but $f=0$. By Proposition~\ref{prop:skew-vandermonde-full-rank}, the $k\times |B|$ matrix $G_{0,B}$ has full rank. Therefore, $(0,h)P=0$ implies $h=0$, which contradicts the assumption that $h\neq 0$. Now suppose that $f\neq 0$ but $h=0$. Similarly, by Proposition~\ref{prop:skew-vandermonde-full-rank}, the $k\times |B|$ matrix $G_{1,B}$ has full rank. Therefore, $(f,0)P=0$ implies $f=0$, which contradicts the assumption that $f\neq 0$, and we are left with the only possibility that $f\neq 0,h\neq 0$ for $(f,h)$ to be nonzero.

The conclusion of this theorem for the case $|A|<k$ will follow from the following claim.
\begin{claim}\label{cl:diff-dim-1}
	If $|A|<k$, $f\neq 0$, and $(f,h)P=0$ then
	\begin{align}
		\dim_{\ff_q}\Span_{\ff_q}\{\cD_{h,1}(\alpha_v)\mid v\in B\} &\leq k,\label{eq:low}\\
		\dim_{\ff_q}\Span_{\ff_q}\{\cD_{f,\gamma}(\beta_v)\mid v\in B\} &\geq k+1.\label{eq:high}
	\end{align}
\end{claim}
Before proving the claim, let us show that this claim indeed implies the theorem. Suppose there exists $f\neq 0,h\neq 0$ such that $(f,h)P=0$. Then Claim~\ref{cl:diff-dim-1} holds. By \eqref{eq:high} there exists $S\subset B$ with $|S|\geq k+1$ such that $\{\cD_{f,\gamma}(\beta_v)\mid v\in S\}$ is linearly independent over $\ff_q$. Moreover, since $(f,h)P=0$ then by \eqref{eq:b} we have $\cD_{f,\gamma}(\beta_v)+\cD_{h,1}(\alpha_v)=0$ for all $v\in S$. It follows that $\{\cD_{h,1}(\alpha_v)\mid v\in S\}$ is linearly independent over $\ff_q$, which contradicts \eqref{eq:low} since $|S|\geq k+1$. Therefore, if $(f,h)P=0$ then $(f,h)=0$, and the theorem follows. It remains to establish Claim~\ref{cl:diff-dim-1}.

\emph{Proof of Claim~\ref{cl:diff-dim-1}}: 
Let us first establish \eqref{eq:low}. Note that we have $|B|>k$. By construction \eqref{eq:alpha} of $\alpha_i$, any $k$ elements of $\{\alpha_v \mid v\in B\}$ form a maximally linearly independent subset of $B$. In other words, there exists $E\subset B$ with $|E|=k$ such that for any $v^*\in B\setminus E$ the element $\alpha_{v^*}$ can be written as an $\ff_q$-linear combination of $\{\alpha_v \mid v\in E\}$.
By Lemma~\ref{le:skew-linearized}, $\cD_{h,1}$ is an $\ff_q$-linear map. Therefore, for any $v^*\in B\setminus E$ the element $\cD_{h,1}(\alpha_{v^*})$ can be written as an $\ff_q$ linear combination of $\{\cD_{h,1}(\alpha_v)\mid v\in E\}$.
Thus, $\dim_{\ff_q}\Span_{\ff_q}\{\cD_{h,1}(\alpha_v)\mid v\in B\} \leq k$.

Next, let us show \eqref{eq:high}. Note that $f(x)$ is a nonzero skew polynomial with $\deg f\leq k-1$. Therefore, by Corollary~\ref{coro:skew-ker-dim},
$\dim_{\ff_q}\ker \cD_{f,1} + \dim_{\ff_q}\ker \cD_{f,\gamma} \leq k-1$. At the same time, by \eqref{eq:a} we have $\cD_{f,1}(\alpha_u)=0$ for all $u\in A$ where $|A|<k$. In addition, by construction \eqref{eq:alpha} of $\alpha_i$, the elements in $\{\alpha_u\mid u\in A\}$ are linearly independent over $\ff_q$. Therefore, we have $\dim_{\ff_q}\ker \cD_{f,1} \geq |A|$. It follows that
\begin{align}
	\dim_{\ff_q}\ker \cD_{f,\gamma} &\leq k-1-\dim_{\ff_q}\ker \cD_{f,1}\nonumber\\
	& \leq k-1-|A|\label{eq:ker-f}
\end{align} 

Suppose $\dim_{\ff_q}\Span_{\ff_q}\{\cD_{f,\gamma}(\beta_v)\mid v\in B\}< k+1$. Then there exists $T\subset B$ with $|T|\leq k$ such that for all $v^*\in B\setminus T$
\begin{align*}
	\cD_{f,\gamma}(\beta_{v^*}) = \sum_{v\in T}\eta_v \cD_{f,\gamma}(\beta_v),
\end{align*} where $\eta_v\in \ff_q$. By Lemma~\ref{le:skew-linearized}, $\cD_{f,\gamma}$ is an $\ff_q$-linear map. Thus, for all $v^*\in B\setminus T$ we have
\begin{align*}
	\cD_{f,\gamma}(\beta_{v^*}) = \cD_{f,\gamma}\Big(\sum_{v\in T}\eta_v\beta_v\Big).
\end{align*}
Let $\bar{\beta}_{v^*}=\beta_{v^*}-\sum_{v\in T}\eta_v\beta_v$. Then $\bar{\beta}_{v^*}\in\ker \cD_{f,\gamma}$. We claim that $\bar{\beta}_{v^*}\notin \Span_{\ff_q}\{\beta_v\mid v\in T\}$. Indeed, if $\bar{\beta}_{v^*}\in\Span_{\ff_q}\{\beta_v\mid v\in T\}$, then $\beta_{v^*}$ can be written as a linear combination of $\{\beta_v\mid v\in T\}$ over $\ff_q$ where $|T|\leq k$. However, by construction \eqref{eq:beta} of $\beta_i$, any $t=2k$ elements of $\{\beta_i \mid i=1,\ldots,n\}$ are linearly independent over $\ff_q$, which is a contradiction. Furthermore, by construction \eqref{eq:beta} of $\beta_i$, the set $\{\bar{\beta}_{v^*}\mid v^*\in B\setminus T\}$ is linearly independent over $\ff_q$. Since $\bar{\beta}_{v^*}\in\ker\cD_{f,\gamma}$ for all $v^*\in B\setminus T$, we have $\dim_{\ff_q}\ker \cD_{f,\gamma}\geq |B\setminus T|\geq |B|-k$. Therefore, from \eqref{eq:ker-f} we obtain $k-1-|A|\geq |B|-k$, which implies $2k-1\geq |A|+|B|=2k$ and leads to a contradiction. This completes the proof of Claim~\ref{cl:diff-dim-1} as well as the proof for the theorem.
\end{proof}

The above result gives an explicit construction of MDP convolutional code with rate $k/n<1/2$ and degree $\delta=k$. It can be easily extended to the case of codes with higher rate and the same degree by applying Theorem~\ref{thm:mdp-dual}.

\begin{corollary}\label{coro:m1-dual}
	Let $n>2k$ and $\cC^\perp$ be the dual code of the $(n,k)$ code $\cC$ defined in Construction~\ref{con:1}.
	Then $\cC^\perp$ is an $(n,n-k,\delta=k)$ MDP convolutional code over $\ff_{q^t}$.
\end{corollary}

Combining Theorem~\ref{thm:m1} and Corollary~\ref{coro:m1-dual}, we have the following result.

\begin{corollary}\label{coro:m1-k}
	Let $n\neq 2k$. There exists a family of $(n,k,\delta=k)$ MDP convolutional codes that can be constructed explicitly over a finite field of size $\Theta(n^{2k})$. 
\end{corollary}

Furthermore, the duality of MDP convolutional codes implies that Construction~\ref{con:1} also gives rise to a family of $(n,k)$ MDP convolutional codes with degree $\delta=n-k$.

\begin{corollary}\label{coro:m1-n-k}
	Let $n\neq 2k$. There exists a family of $(n,k,\delta=n-k)$ MDP convolutional codes that can be constructed explicitly over a finite field of size $\Theta(n^{2(n-k)})$.
\end{corollary}

\begin{proof}
	Let $\tilde{k}=n-k$ and $\tilde{\cC}$ be the $(n,\tilde{k})$ convolutional code constructed from Construction~\ref{con:1}. By theorem~\ref{thm:m1}, if $n>2\tilde{k}$ then $\tilde{\cC}$ is an $(n,\tilde{k},\delta=\tilde{k})$ unit memory MDP convolutional code over $\ff_{q^{2\tilde{k}}}$. Furthermore, let $\tilde{\cC}^\perp$ be the dual code of $\tilde{\cC}$. Then it follows from Corollary~\ref{coro:m1-dual} that $\tilde{\cC}^\perp$ is an $(n,n-\tilde{k},\delta=\tilde{k})$ MDP convolutional code over $\ff_{q^{2\tilde{k}}}$.
\end{proof}

\section{Concluding remarks}
\label{sec:con}

A few observations can be made from Corollary~\ref{coro:m1-k} and \ref{coro:m1-n-k}.	As mentioned before, it is desirable for $(n,k)$ convolutional codes to have small degree. Corollary~\ref{coro:m1-k} and \ref{coro:m1-n-k} imply that, given any parameter $n$ and $k$ such that $n\neq 2k$, there exists a family of $(n,k)$ MDP convolutional codes with degree $\delta=\min\{k,n-k\}$ that can be constructed explicitly over a finite field of size $\Theta(n^{2\delta})$.
In particular, if $R=k/n$ is a constant such that $0<R<1,R\neq1/2$, then the codes can be constructed over a finite field of size $2 ^{\Theta(\delta\log\delta)}$. A comparison of the field size requirement for constructions of MDP convolution codes with constant rate and memory is given in Table~\ref{tab:comp}.

\begin{table}[!t]
	\renewcommand{\arraystretch}{1.3}
	\centering
	\begin{tabular}{|c|c|}
		\hline
		Field size  & Reference\\
		\hline\hline
		$2^{O(\delta^2)}$ & \cite{gluesing2006strongly}\\
		\hline
		$2^{2^{O(\delta)}}$ & \cite{almeida2013new}\\
		\hline
		$2^{O(\delta^2\log\delta)}$ & \cite{alfarano2020weighted}\\
		\hline
		$2^{O(\delta\log\delta)}$ & This paper\\
		\hline\hline
		$2^{O(\delta)}$ & Conjecture in \cite{hutchinson2008superregular}\\
		\hline
	\end{tabular}
	\caption{A comparison of the field size requirement for constructions of MDP convolutional codes with constant rate and memory.}
	\label{tab:comp}
\end{table}

It is worth noting that by Definition~\ref{def:sMDS-MDP}, when $n<2k$ we have $M=L=1$, and thus in that case the constructions given by Corollary~\ref{coro:m1-k} and \ref{coro:m1-n-k} are also strongly-MDS convolutional codes.

Finally, it is interesting to generalize the approach in this paper to construct MDP convolutional codes for any $L>1$ and extend the result to any value of $\delta$.

\section*{Acknowledgements}
The author is grateful to Alexander Barg for discussion and comments on the first version of this paper.
	\bibliographystyle{IEEEtran}
	\bibliography{MDP}
\end{document}